\documentclass[final, nomarks]{dmtcs-episciences}

\usepackage[utf8]{inputenc}
\usepackage{subfigure}

\usepackage{amsthm,amssymb,amsfonts,amsmath}
\newtheorem{theorem}{Theorem}
\newtheorem{lemma}[theorem]{Lemma}
\newtheorem{corollary}[theorem]{Corollary}

\theoremstyle{definition}
\newtheorem{definition}[theorem]{Definition}

\usepackage{graphicx}
\graphicspath{{figures/}}

\newcommand{\fig}[1]{\figurename~\ref{#1}}

\makeatletter
\g@addto@macro\bfseries{\boldmath}
\makeatother

\title{Planar 3-SAT with a Clause/Variable Cycle}

\author{Alexander Pilz\affiliationmark{1}\thanks{Supported by a Schr\"odinger fellowship of the Austrian Science Fund (FWF): J-3847-N35.
This work was done while the author was at the Department of Computer Science of ETH Z\"urich.}}
\affiliation{Institute of Software Technology,
Graz University of Technology, Austria.}

\received{2018-8-7}

\revised{2019-1-24, 2019-4-30}

\accepted{2019-5-7}

\publicationdetails{21}{2019}{3}{18}{4742}

\begin{document}

\maketitle

\begin{abstract}
In the \textsc{Planar 3-SAT} problem, we are given a 3-SAT formula together with its incidence graph, which is planar, and are asked whether this formula is satisfiable.
Since Lichtenstein's proof that this problem is NP-complete, it has been used as a starting point for a large number of reductions.
In the course of this research, different restrictions on the incidence graph of the formula have been devised, for which the problem also remains hard.

In this paper, we investigate the restriction in which we require that the incidence graph can be augmented by the edges of a Hamiltonian cycle that first passes through all variables and then through all clauses, in a way that the resulting graph is still planar.
We show that the problem of deciding satisfiability of a 3-SAT formula remains NP-complete even if the incidence graph is restricted in that way and the Hamiltonian cycle is given.
This complements previous results demanding cycles only through either the variables or clauses.

The problem remains hard for monotone formulas, as well as for instances with exactly three distinct variables per clause.
In the course of this investigation, we show that monotone instances of \textsc{Planar 3-SAT} with exactly three distinct variables per clause are always satisfiable, thus settling the question by Darmann, D\"ocker, and Dorn on the complexity of this problem variant in a surprising way.
\end{abstract}

\section{Introduction}
Let $\phi$ be a Boolean formula in conjunctive normal form (CNF) and let $G_\phi$ be a graph whose vertices are the variables and the clauses of $\phi$ such that (1) every edge of $G_\phi$ is between a variable and a clause and (2) there is an edge between a variable~$v$ and a clause $c$ if and only if $v$ occurs in $c$ (negated or unnegated).
We call $\phi$ a \emph{CNF formula} and $G_\phi$ is called the \emph{incidence graph} of $\phi$.
A CNF formula is a \emph{3-SAT formula} if every clause contains at most three variables.
(We will also discuss the case where every clause contains exactly three distinct variables.)
The \textsc{Planar 3-SAT} problem asks whether a given 3-SAT formula~$\phi$ is satisfiable, given that $G_\phi$ is a planar graph.
This problem has been shown to be NP-complete by Lichtenstein~\cite{lichtenstein}.
(In contrast to the general version, a PTAS is known for maximizing the number of satisfied clauses for the planar version of the 3-SAT problem~\cite{ptas}.)
See \figurename~\ref{fig_incidence_both} for drawings of an incidence graph.

\begin{figure}
\centering
\includegraphics[width=\columnwidth]{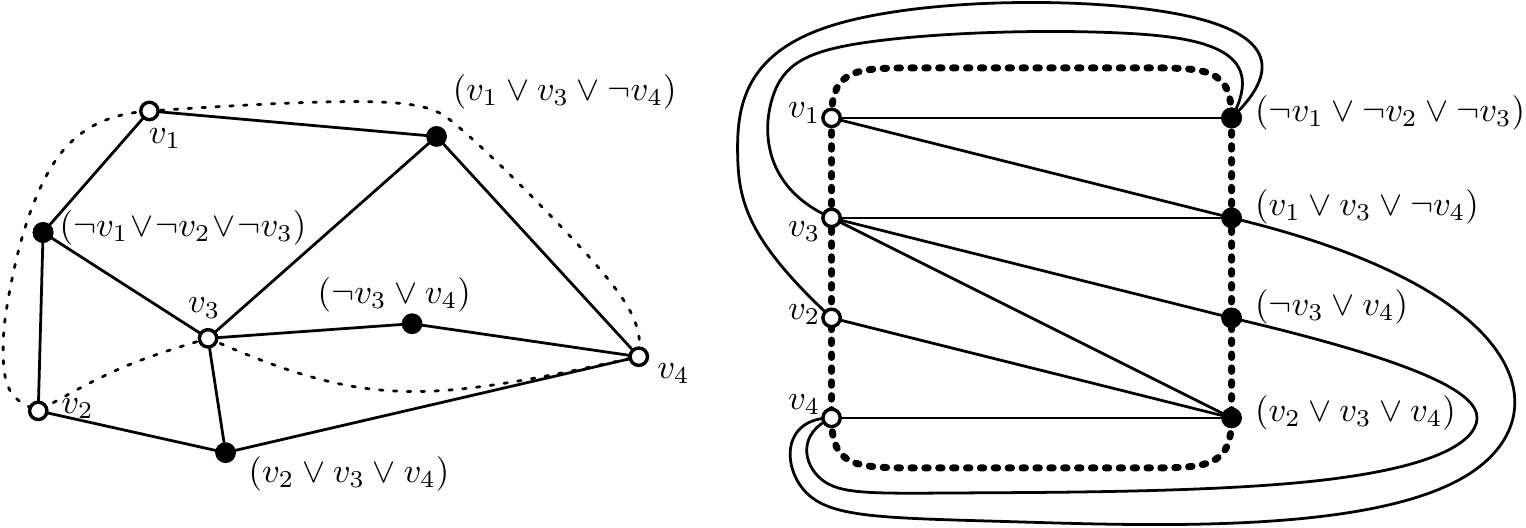}
\caption{
Left: An incidence graph for a 3-SAT formula that is planar.
Clauses are black vertices and variables are white vertices.
The graph is augmented by a spanning cycle (dotted) on the variable vertices.
Right: A \textsc{Linked Planar 3-SAT} instance, with a Hamiltonian cycle $\kappa$ (dotted) passing through the variables and the clauses.
(At most two edges of $\kappa$ may already be present in the incidence graph.)
}
\label{fig_incidence_both}
\end{figure}

Reducing from \textsc{Planar 3-SAT} is a standard technique to show NP-hardness of problems in computational geometry.
In these reductions, the vertices and edges of $G_\phi$ are replaced by gadgets (consisting of geometric objects) that influence each other.
However, it is often useful to have further restrictions on $G_\phi$ or on how $G_\phi$ can be embedded.

Lichtenstein's reduction already contains such a restriction: the problem remains NP-complete even if the graph remains planar after adding a cycle whose vertices are exactly the variables of the formula~\cite{lichtenstein}, and this cycle is part of the input.
We call it a \emph{variable cycle}.
(Note that this does not change the formula $\phi$, but merely says that an algorithm can expect additional structure on the incidence graph~$G_\phi$.)
This fact allows for placing the variable vertices along a line, connected to ``three-legged'' clauses above and below that line (stated more explicitly in~\cite{knuth_ortho}).
For later reference, we make this notion more explicit.
In a \emph{three-legged embedding} of an incidence graph, variables are represented by disjoint line segments on the $x$-axis, and clauses are represented by points;
if a variable occurs in a clause, the variable's segment is connected to the clause's point by a curve consisting of a vertical line segment and at most one horizontal line segment, such that no two such edges intersect (except at their endpoints at a shared vertex); see \figurename~\ref{fig_three_legs}.

De Berg and Khosravi~\cite{de_berg} showed that it is also possible to have all literals in the clauses above the line to be positive, and all literals in clauses below the line negative.
In Lichtenstein's reduction, one may as well add a \emph{clause cycle} whose vertices are exactly the clauses of the formula while keeping the graph planar~\cite{clause_cycle}.

We call the \textsc{Planar 3-SAT} variants in which we require a variable cycle and a clause cycle \textsc{Var-Linked Planar 3-SAT} and \textsc{Clause-Linked Planar 3-SAT}, respectively.%
\footnote{While Lichtenstein's definition of planar 3-SAT~\cite{lichtenstein} already requires the cycle through the variables, this property is often considered an explicit restriction.
We thus follow the terminology of Fellows et al.~\cite{fellows} to emphasize when the variable cycle is needed.}
Even though the incidence graphs constructed by Lichtenstein can be augmented with both a clause cycle~\cite{clause_cycle} and a variable cycle, one cannot adapt Lichtenstein's construction to always obtain both cycles without any crossings.%
\footnote{If not stated otherwise, we will implicitly require the incidence graphs augmented by the additional edges to be planar throughout this paper.}

\begin{figure}
\centering
\includegraphics[page=3]{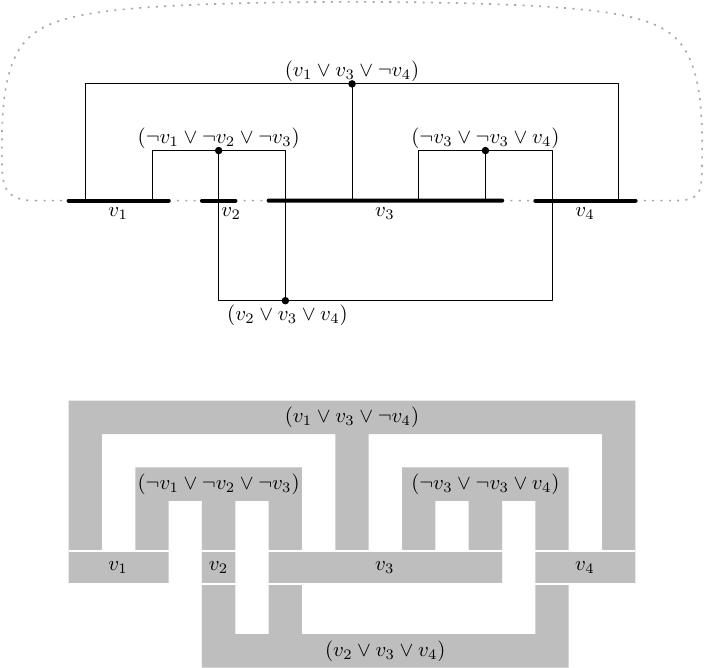}
\caption{
A three-legged embedding of a \textsc{Planar 3-SAT} instance with variables on a line similar to~\cite[p.~425]{knuth_ortho} (top).
(There, two-variable clauses are transformed to three-variable clauses that contain one literal twice, a construction that is not necessary but possible for the initial graph in our reduction.)
An augmentation by a cycle through the variable vertices is shown dotted.
Usually, the elements of the drawing are ``thickened'' to allow for the embedding of gadgets (bottom).
}
\label{fig_three_legs}
\end{figure}

Our research is motivated by the following problem that attempts to combine these restrictions.
See \figurename~\ref{fig_incidence_both}~(right) for an accompanying illustration.

\begin{definition}[\textsc{Linked Planar 3-SAT}]
Let $G_\phi = (C \cup V, E)$ be the incidence graph of a 3-SAT formula~$\phi$, where $C$ is the set of clauses and $V$ is the set of variables of~$\phi$.
Further, let $\kappa$ be a Hamiltonian cycle of $C \cup V$ that first visits all elements of $C$ and then all elements of $V$.
Suppose that the union of $G_\phi$ and $\kappa$ is a planar graph.
The \textsc{Linked Planar 3-SAT} problem asks, given $\phi$, $G_\phi$, and $\kappa$, whether $\phi$ is satisfiable.
\end{definition}

Related problems in which all variables or all clauses can be drawn incident to the unbounded face are known to be in~P, due to results by Knuth~\cite{nested}, and Kratoch\'{v}il and K\v{r}iv{\'{a}}nek~\cite{co_nested}, respectively.
In particular, this is the case when there is a variable cycle and a path connecting all clauses or vice versa.
Also, observe that if there is a variable cycle and at each variable the edges to clauses emanate to the same side of the cycle, then the incidence graph can be re-drawn with the variables on the unbounded face (after removing the variable cycle);
the analogous holds for a clause cycle.
One way of tackling the \textsc{Linked Planar 3-SAT} problem could be to show that $G_\phi$ has bounded treewidth.
For such instances, the satisfiability of $\phi$ can be decided in polynomial time~\cite{counting_truth_assignments}.
(This generalizes the above-mentioned results, as every $k$-outerplanar graph has treewidth at most $3k-1$~\cite{bodlaender}\footnote{An outerplanar graph is 1-outerplanar, and a $k$-outerplanar graph is $(k-1)$-outerplanar after removing the vertices incident to the unbounded face.}; see also Demaine's lecture notes~\cite{demaine_lecture}.)
However, with the right perspective on the \textsc{Linked Planar 3-SAT} problem, it will be easy to observe that there are formulas whose incidence graph has a grid minor with a linear number of vertices (and thus such graphs have unbounded treewidth).
It is the same perspective through which we will show NP-completeness of the problem in Section~\ref{sec_paths}, using a reduction from \textsc{Planar 3-SAT}.
We note that requiring an arbitrary Hamiltonian cycle is not a restriction:
As the incidence graph is bipartite, it is known that its page number is two~\cite{left_first}.%
\footnote{The \emph{page number} of a graph is $k$ if it can be drawn without crossings on $k$ pages (i.e., half-planes) of a book with all vertices on the book's spine (i.e., the line bounding all half-planes).}
Hence, we can always add a Hamiltonian cycle through the variables and clauses in a planar way (possibly re-using edges of the incidence graph).

\subsection{Motivation}
Restrictions on the problem to reduce from can make NP-hardness reductions simpler.
For reductions from \textsc{Planar 3-SAT}, it is common to actually reduce from \textsc{Var-Linked Planar 3-SAT}, using the variable cycle, in particular the ``three-legged'' embedding of~\cite{knuth_ortho}.
Also, the clause cycle has been used~\cite{clause_cycle,fellows,few_lines}.
For an exhaustive survey on the numerous variants of \textsc{Planar 3-SAT}, see the thesis of Tippenhauer~\cite{tippenhauer}.

One motivation for considering \textsc{Linked Planar 3-SAT} is the framework for showing NP-hardness of platform games by Aloupis et al.~\cite{nintendo}.
In this class of reductions from 3-SAT to such games, a player's character starts at a specified position and traverses all variable gadgets, making a decision on their truth value.
Clauses connected to the satisfied literal can be ``unlocked'' by visiting these clauses.
Finally the player's character has to traverse all clause gadgets to reach the finish (called the ``check path'').
The framework then requires a game-specific implementation of the gadgets for start, finish, variables, clauses, and crossovers.
Reducing from \textsc{Linked Planar 3-SAT} removes this dependency on crossover gadgets.
(In particular, Theorem~\ref{thm_side} can be used to show that the traversal through the variables can be done without crossings.)
See \cite[Section~2.1]{nintendo} for a more detailed description.

\subsection{Results}
A summary of presented results, as well as results about related problems, is given in Table~\ref{tbl:overview}.
We first prove that \textsc{Linked Planar 3-SAT} is NP-complete.
In the following sections, we refine the construction to show that restricted variants of the problem remain hard as well.
In particular, we do this for \textsc{Monotone Planar 3-SAT}; here, every clause contains either only negated literals or only unnegated ones.
Further, we show hardness for a variant in which the edges leaving a clause to negated variables are all on the same side of~$\kappa$ (recall that $\kappa$ is the Hamiltonian cycle that first visits all variables and then all clauses), and edges to unnegated ones are on the other side.
Also, we may require that all clauses contain exactly three distinct variables.
A cycle $\kappa$ in instances of \textsc{Positive Planar 1-in-3-SAT} (which requires exactly one true literal in each clause, and no literal is negated) also keeps the problem hard.
Table~\ref{tbl:overview} gives an overview of our results and related work.

Finally, we discuss settings in which the planarity constraint is fulfilled only by satisfiable formulas.
In particular, we show that planar CNF formulas with at least four variables per clause are always satisfiable.
The same holds for instances of \textsc{Monotone Planar 3-SAT} with exactly three variables per clause.
This solves an open problem by Darmann, D\"ocker, and Dorn~\cite{darmann_planar,darmann_journal}, who show that the corresponding problem with at most three variables per clause remains NP-complete with bounds on the variable occurrences, which refines the result of de Berg and Khosravi~\cite{de_berg}.

\begin{table}
\footnotesize
\centering
\begin{tabular}{|p{13em}|c|c|c|c|}
\hline
 & General: & Var-linked: & Clause-linked: & Linked: \\ \hline \hline
\textsc{Planar 3-SAT} & NP-c.~\cite{lichtenstein} & NP-c.~\cite{lichtenstein} & NP-c.~\cite{clause_cycle} & NP-c.~(Thm.~\ref{thm:linked_planar_3_sat}) \\ \hline
\textsc{Planar 3-SAT} with 3 literals per clause & NP-c.~\cite{mansfield} & NP-c.~~(Sec.~\ref{sec:different_cycles}) & NP-c.~(Sec.~\ref{sec:different_cycles}) & NP-c.~(Thm.~\ref{thm:linked_three_distinct}) \\ \hline
\textsc{Monotone Planar 3-SAT} & NP-c.~\cite{de_berg} & NP-c.~\cite{de_berg} & NP-c.~(Sec.~\ref{sec:different_cycles}) & NP-c.~(Thm.~\ref{thm:monotone_linked}) \\ \hline
\textsc{Positive Planar 1-in-3-SAT} & NP-c.~\cite{mulzer_rote} &  NP-c.~\cite{mulzer_rote} &  NP-c.~\cite{few_lines} &  NP-c.~(Thm.~\ref{thm:planar_1_in_3}) \\ \hline
\textsc{Planar 3-SAT} with negated edges on one side & n/a &  NP-c.~(Sec.~\ref{sec:different_cycles}) &  NP-c.~(Sec.~\ref{sec:different_cycles}) &  NP-c.~(Thm.~\ref{thm_side}) \\ \hline
\textsc{Planar SAT} with three negated/unnegated per clause & P~(Thm.~\ref{thm:three_distinct_satisfiable}) & & & \\ \hline
\textsc{Planar SAT} with $\geq 4$ variables per clause & P~(Thm.~\ref{thm:four_satisfiable}) & & & \\ \hline
\end{tabular}
\caption{Summary of results presented in this paper and related problems.}
\label{tbl:overview}
\end{table}

\section{NP-hardness of \textsc{Linked Planar 3-SAT}}\label{sec_paths}
We now show that \textsc{Linked Planar 3-SAT} is NP-hard by reducing \textsc{Planar 3-SAT} to it.
We are thus given a 3-SAT formula $\tilde \phi$ with variable set $\tilde V$, clause set $\tilde C$, and incidence graph $G_{\tilde \phi}$.
Our goal is to construct another formula $\phi$ that is an instance of \textsc{Linked Planar 3-SAT} with incidence graph $G_{\phi}$ and Hamiltonian cycle $\kappa$ such that $\phi$ is satisfiable if and only if $\tilde \phi$ is.
The construction of $G_{\phi}$ will be given by an embedding.

We start with a suitable embedding of the initial graph $G_{\tilde \phi}$.
We produce an embedding~$\Gamma$ of $G_{\tilde \phi}$ on the integer grid (i.e., all vertices have integer coordinates), with straight-line edges.
(It seems easier to describe the construction with a fixed straight-line embedding of $G_{\tilde \phi}$ on the integer grid.
However, this is only one possible variant, as discussed at the end of this section.)
The crucial property we use is that any vertical line intersects an edge at most once (i.e., each edge is an $x$-monotone curve);
this is of course fulfilled by a straight-line embedding.
We further require that the size of the grid is polynomial in the size of~$\tilde \phi$.
It is well-known that a planar graph with $n$ vertices can be embedded with straight-line edges on a $O(n) \times O(n)$ grid in $O(n)$ time~\cite{planar_drawing,schnyder}.
We can take such an embedding and perturb the clause vertices s.t.\ each variable has an even $x$-coordinate, and the $x$-coordinate of each clause is odd.
This can be done without introducing crossings by choosing the grid sufficiently large and scaling $\Gamma$; a blow-up by a factor polynomial in~$n$ is sufficient.%
\footnote{We can consider the smallest horizontal distance between a vertex and a non-incident edge.
This distance $v$ is rational with numerator and denominator quadratic in the largest coordinate.
By multiplying the $x$-coordinates of the vertices by $2/v$, rounding, and again multiplying by 2, we get the desired embedding, similar to~\cite[Lemma~6.1]{simple_flip_hard}.
But again, we emphasize that our use of a straight-line embedding is for convenience only and can be replaced in a way discussed at the end of this section.}
More specifically, we scale $\Gamma$ by a polynomial multiple of 2 and increase the $x$-coordinate of each clause vertex by~1.

With a suitable drawing of $G_{\tilde \phi}$ at hand, we start the drawing of $G_\phi$ with the curve that will contain the cycle $\kappa$ (we add its vertices later).
It can be partitioned into two paths, one that will contain the elements of the variable set $V$ ($\kappa_V$) and one for the elements of the clause set ($\kappa_C$).
In our drawing shown in~\figurename~\ref{fig_grid_region}, we obtain a rectangular region $R$, whose intersection with $\kappa$ consists of vertical line segments of unit distance, in alternation belonging to $\kappa_V$ and $\kappa_C$.
We call them the \emph{clause segments} and \emph{variable segments}, respectively.
We assume that the segments are placed on the integer grid with unit distance, with the variable segments having even $x$-coordinates and the clause segments having odd $x$-coordinates.
In other words, $R$ represents a grid in which the columns are traversed in alternation by $\kappa_V$ and $\kappa_C$.

\begin{figure}
\centering
\includegraphics{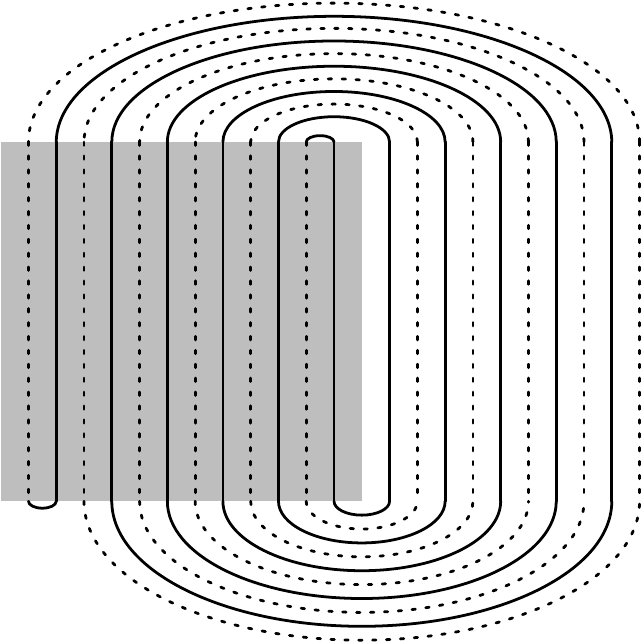}
\caption{Drawings of the paths $\kappa_C$ (solid) and $\kappa_V$ (dotted), containing the clauses and variables of $\phi$, respectively.
They intersect a rectangle $R$ (gray) in vertical segments.
The incidence graph is drawn inside~$R$.
}
\label{fig_grid_region}
\end{figure}

Place $\Gamma$ inside $R$.
By the construction of $\Gamma$, all its variable vertices have even $x$-coordinates and are thus on variable segments, and its clause vertices are on clause segments.
Throughout this paper, we will make use of this construction in various proofs, and thus call it an \emph{$R$-drawing} for future reference.
We obtain $G_\phi$ by replacing the edges of $\Gamma$ by gadgets, consisting of subgraphs of~$G_\phi$.
In the construction, we will make use of ``cyclic implications'', effectively copying the value of a variable;
the graph $G_\phi$ will contain many pairs of variables $x$ and $x'$ with a clause $c_x = (\neg x \lor x')$ and a clause $c'_x = (x \lor \neg x')$.
Clearly, $x = x'$ in any satisfying assignment.
We depict the negation in such a clause $c_x$ by an arrow from $x$ to $c_x$, and from $c_x$ to $x'$ (where the arrows are also edges of the incidence graph).
In general, we use the convention that an arrow from a variable to a clause denotes that the variable occurs negated in that clause, while an arrow from the clause to the variable means that the variable occurs unnegated.

We replace each edge~$e$ of $\Gamma$ by a sequence of so-called \emph{connector gadgets}.
A connector gadget consists of two variables $x$ and $x'$, and two clauses $c_x = (\neg x \lor x')$ and $c_x' = (\neg x' \lor x)$ (implying $x = x'$ in any satisfying truth assignment).
The variable vertices $x$ and $x'$ are placed on the intersections of $e$ with two consecutive variable segments in $R$, and the clause vertices $c_x$ and $c_x'$ are placed on the clause segment between them (also close to the crossing of $e$ and the clause segment).
Note that these new vertices subdivide~$\kappa$.
An edge in $\Gamma$ connecting a variable $\tilde v \in \tilde V$ to a clause $\tilde c \in \tilde C$ that crosses $\kappa$ can be replaced by a sequence of connector gadgets in a sufficiently small neighborhood of the edge, as shown in \figurename~\ref{fig_edge_replacement}.
Thus, in the resulting drawing, we have subdivided $\kappa$ to remove crossings with $e$, and the resulting formula is satisfiable if and only if the initial formula is satisfiable.

\begin{figure}
\centering
\includegraphics{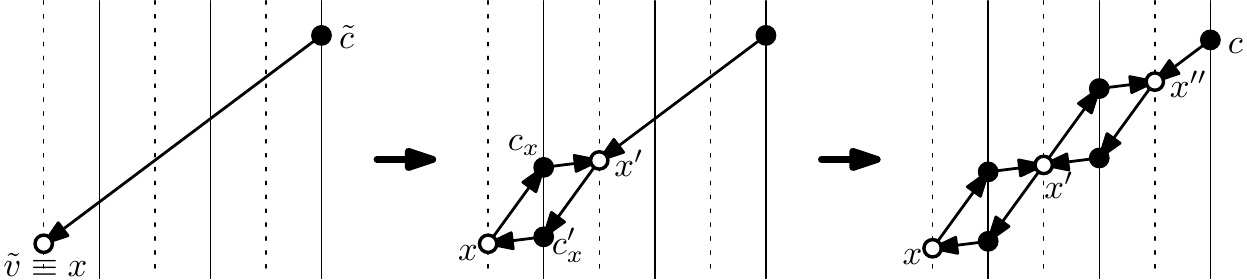}
\caption{
An edge of the initial incidence graph crossing the cycle $\kappa$ (left) can be replaced by a sequence of connector gadgets.
Variable vertices are white dots, clause vertices are black.
The first connector gadget is attached to a variable $\tilde v$, i.e., $\tilde v$ plays the role of variable~$x$ in the connector gadget.
The two variables in the connector gadget (middle) have the same value, and the occurrence of a variable $x$ in the clause $\tilde c$ is replaced by the last variable~$x''$ to obtain an equivalent clause~$c$.
}
\label{fig_edge_replacement}
\end{figure}

Replacing the edges of $\Gamma$ by the connector gadgets results in a drawing of $G_\phi$.
As all edges are $x$-monotone, they intersect $\kappa$ as required for the construction of the edges (and in particular each edge crosses clause segments and variable segments in alternation).
Thus, all crossings with $\kappa$ are replaced.
Thus this drawing is planar and contains a number of vertices that is polynomial in $|C|$, as the number of crossings of an edge in $\Gamma$ with clause and variable segments (and thus the number of vertices needed to replace the edge) is bounded by the grid size.
(For a more precise estimate of the number of vertices see the discussion at the end of this section.)
Also, all ``new'' clauses contain only two variables (we will see a modification of the reduction without this property).
As none of the edges in the resulting embedding of $G_\phi$ crosses the initially drawn cycle for~$\kappa$, $G_\phi$ can be augmented by $\kappa$ along that cycle maintaining planarity.
Finally, observe that $\phi$ is satisfiable if and only if $\tilde \phi$ is:
the two variables of a connector gadget must have the same value, and the clauses not part of the connector gadget are the clauses of $\tilde C$ in which we replaced variables by others that have to be equal.
We thus obtain our main result.

\begin{theorem}\label{thm:linked_planar_3_sat}
\textsc{Linked Planar 3-SAT} is NP-complete.
\end{theorem}

Observe that our construction also does not change the number of satisfying assignments, i.e., the reduction is parsimonious.
Since counting the number of satisfying assignments to a planar 3-SAT formula is \#P-complete~\cite{hunt}, this also holds for \textsc{Linked Planar 3-SAT}.

\begin{theorem}
\textsc{Linked Planar 3-SAT} is \#P-complete.
\end{theorem}

\paragraph*{A note on the embedding.}
In our reduction, we were using a straight-line drawing of the incidence graph for convenience.
However, as the reader may have noticed by now, we do not rely on an actual drawing in the plane.
We are merely considering the topology of the embedding to add the edges of $\kappa$ and the connector gadgets in order to construct a planar graph.
The clause and variable segments in the region $R$ can be seen as snapshots of a vertical sweep line passing over the drawing.
With this in mind, we can get an $R$-drawing of a graph by a topological sweep of the incidence graph.
We can think of a bi-infinite curve~$K$ that sweeps over the drawing and that at any time during the sweep intersects each edge in at most one point.
The sweep curve is implemented as a cut of the graph.
See, e.g.,~\cite{guibas_seidel}.
This sweep can be done such that the curve never intersects a clause vertex and a variable vertex at the same time.
As illustrated in \fig{fig_sweep}, such a sweep gives a way of adding the clause and variable segments of the reduction, resulting in an $R$-drawing.
When passing over a clause vertex, we make the intersection of $K$ and $R$ a clause segment (that may no longer be straight-line), and to a variable segment when passing over a variable vertex.
If there are, say, two consecutive variable segments, we place a clause segment corresponding to a state of $K$ in-between these segments.
This way of producing the graph $G_\phi$ allows for a more fine-grained analysis of the number of additional vertices required.
We get at most $2(|\tilde V| + |\tilde C|)-1$ clause and variable segments.
As the number of edges in the initial incidence graph $G_{\tilde \phi}$ is linear in the number of vertices, the number of vertices in the final incidence graph is in $O(|\tilde C|^2)$.

\begin{figure}
\centering
\includegraphics{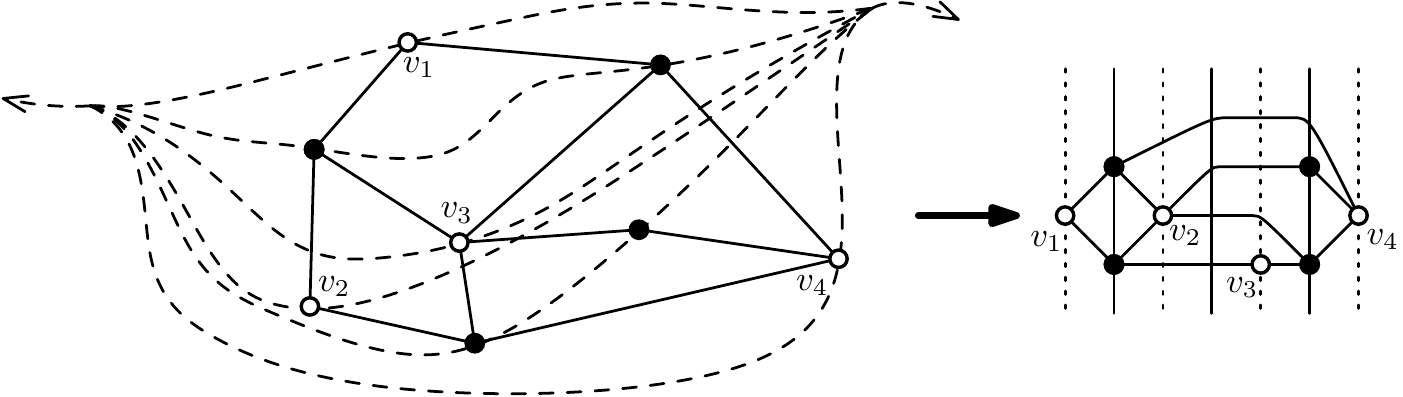}
\caption{A topological sweep of a drawing (left) determines a way of adding clause segments and variable segments for the reduction.}
\label{fig_sweep}
\end{figure}

Let us illustrate the reduction explicitly starting with a \textsc{Var-Linked Planar 3-SAT} instance, not relying on the straight-line embedding.
\figurename~\ref{fig_construction} shows an example of the reduction starting with the three-legged drawing shown in \figurename~\ref{fig_three_legs}.
We can multiply the coordinates of the vertices (and of the bends) by four; this gives enough freedom to embed the vertices accordingly.
Observe that the vertices and bends of such an embedding have coordinates of absolute value at most $4 \cdot 3|\tilde C|$ (and in this representation, observe that we only need to blow the grid up by a factor of~2).
We replace each horizontal bar representing a variable by a sequence of connector gadgets as shown in \figurename~\ref{fig_construction}.
Note that all variables in the gadget have the same value due to the cyclic implication, and all its vertices have distance at most one from the original segment.
The original vertex representing a clause, say, $\tilde c = (v_1 \lor v_3 \lor \neg v_4)$ is replaced by a clause $c = (v_1' \lor v_3' \lor \neg v_4')$.
Observe that due to the cyclic implication, the variables $v_1, v_2, \dots$ need to have the same values as their surrogates $v_1',v_2',\dots$.
Since all coordinates of $\Gamma$ are multiples of four, the gadgets can be drawn inside the ``thickened'' legs of the initial drawing, and thus the new drawing remains crossing-free.

\begin{figure}
\centering
\includegraphics[page=2]{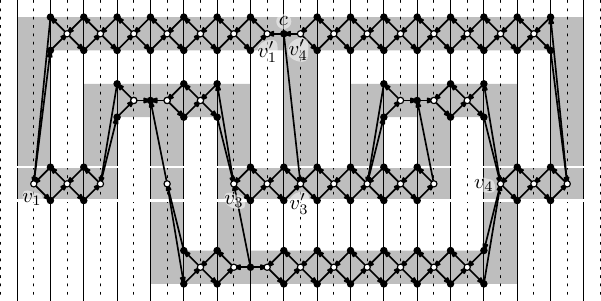}
\caption{A construction based on the three-legged drawing in \figurename~\ref{fig_three_legs}.
A clause $\tilde c = (v_1 \lor v_3 \lor \neg v_4)$ of $\tilde \phi$ in $\Gamma$ is replaced by the clause $c = (v_1' \lor v_3' \lor \neg v_4')$.
The truth value of the variables on the $x$-axis is ``transported'' to the new clause by the according cyclic implications.}
\label{fig_construction}
\end{figure}

\section{Further variants}
We use the main idea of drawing the cycle~$\kappa$ as shown in \fig{fig_grid_region} to obtain similar reductions for variants of the planar satisfiability problem.
If the initial problem is known to be hard, we merely need to find an appropriate connector gadget to replace the crossings of edges of~$G_{\tilde \phi}$ and~$\kappa$.

\subsection{Positive Planar 1-in-3-SAT}
It is easy to transform our reduction to be an instance of \textsc{1-in-3-SAT}, where exactly one variable is true:
we get the ``main'' clauses directly from an initial \textsc{1-in-3-SAT} instance, and the connector gadgets are the same.
We therefore impose two more requirements on the instance.
First, all clauses should have three elements, and second, all literals should be unnegated, i.e., positive.

\begin{definition}[\textsc{Positive Planar 1-in-3-SAT}~\cite{mulzer_rote}]
Given a formula $\phi$ in which each clause contains exactly three distinct unnegated literals, and an embedding of the incidence graph $G_\phi$, the \textsc{Positive Planar 1-in-3-SAT} problem asks whether there exists a satisfying assignment of $\phi$ such that exactly one variable in each clause is true.
\end{definition}

Mulzer and Rote~\cite{mulzer_rote} show that deciding satisfiability of \textsc{Positive Planar 1-in-3-SAT} instances is hard, even if the incidence graph can be augmented by a variable cycle.

\begin{theorem}\label{thm:planar_1_in_3}
The \textsc{Positive Planar 1-in-3-SAT} problem remains NP-complete even for problem instances that are also instances of the \textsc{Linked Planar 3-SAT} problem.
\end{theorem}
\begin{proof}
The reduction starts again by extending the embedding of the initial \textsc{Positive Planar 1-in-3-SAT} incidence graph by clause and variable segments, i.e., by constructing an $R$-drawing.
In our initial reduction, the connector gadget had a width of two, effectively connecting two variable vertices on neighboring variable segments.
In the current reduction, our connector gadgets have a width of four, and it makes two variables equal that have one variable segment strictly in-between.
Hence, we start off with an $R$-drawing as in our initial reduction and then place two pairs of clause and variable segments between each variable segment and the clause segment to its right, and one pair of a variable and a clause segment between a clause segment and the variable segment to its right (obtaining again an $R$-drawing).
See \fig{fig_extend_r_drawing}.
This can be done by an according topological sweep of the initial embedding, as described in Section~\ref{sec_paths}.
(Alternatively, we can use a straight-line embedding on the integer grid.)
By this construction, all variable and clause vertices remain on the appropriate part of $\kappa$.
Further, the number of clause segments strictly between a variable and a clause to its right is even, 
and is odd for any variable to the right of the clause. 

\begin{figure}
\centering
\includegraphics[page=2]{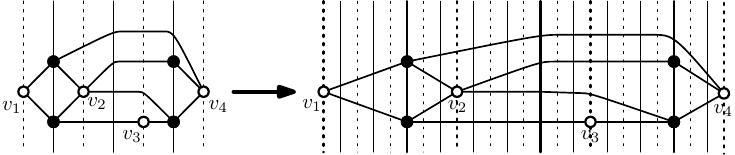}
\caption{An $R$-drawing can be extended by adding further clause and variable segments between the existing one.
}
\label{fig_extend_r_drawing}
\end{figure}

We use two clauses sharing two variables to produce a connector gadget, as shown in \figurename~\ref{fig_1_in_3}.
The two clauses $(x, a, b)$ and $(x', a, b)$ ensure that $x = x'$.
The construction works analogous to the simpler connector gadget of the previous section,%
\footnote{We thank an anonymous referee for simplifying the previously-used gadget.}
except for the following caveat.
The connector gadgets have width 4.
Consider a variable vertex $x$ and a neighboring clause $\tilde c$.
By our construction, if $x$ is to the left of $\tilde c$, the number of pairs of clause and variable segments between $x$ and $\tilde c$ is even.
Hence, the connector gadgets can be placed as shown at the top of \fig{fig_1_in_3}.
If $x$ is to the right of $\tilde c$, we have a similar situation, but one clause segment less, and we place the connector gadget at $\tilde c$ as shown at the bottom of \fig{fig_1_in_3}.
If all of the three variable vertices incident to a clause are to its right, one of the three connector gadgets has to be constructed using the variable segment to the left of the clause, as shown in the sketch on the right of \figurename~\ref{fig_1_in_3}.
(Note that in this case the three connector gadgets incident to the clause can be drawn as shown because there is an odd number of clause segments between each variable and the clause;
also this does not interfere with other gadgets due to the extra clause and variable segments added to the initial $R$-drawing.)
Finally, observe that the number of vertices in the constructed incidence graph is again quadratic in the number of vertices of the initial one by the same arguments as in Section~\ref{sec_paths}.
\end{proof}

\begin{figure}
\centering
\includegraphics{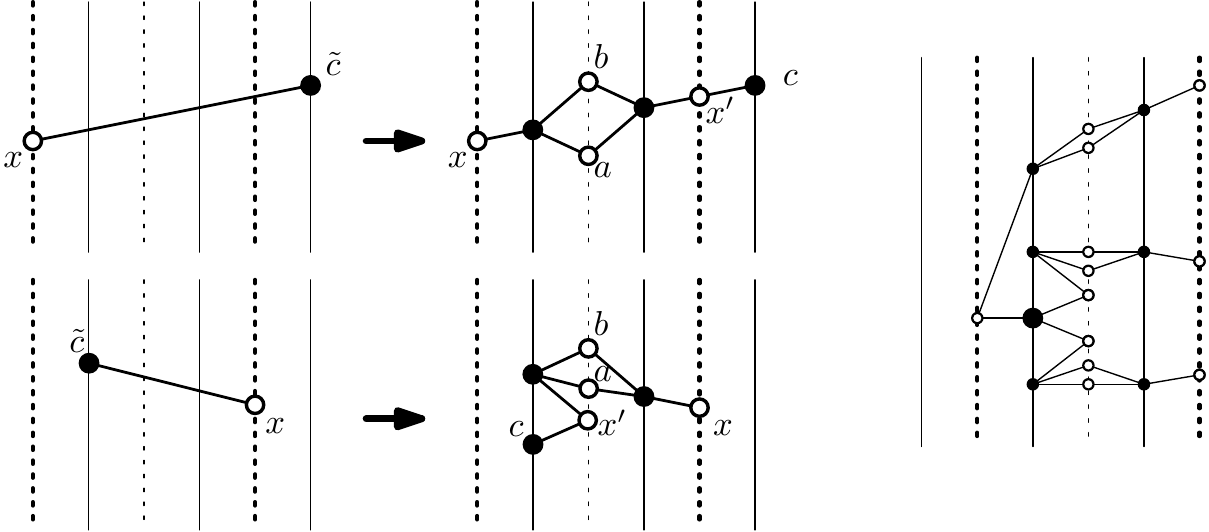}
\caption{Reduction for the \textsc{Positive Planar 1-in-3-SAT} variant.
If $x=1$, then both $a = 0$ and $b = 0$, and thus $x' = 1$ as well.
Otherwise, exactly one of $a$ and $b$ is $1$ and we again have $x = x'$.
When, in the initial graph, the variable vertices incident to a clause have larger $x$-coordinate, the gadgets cannot completely follow the edges, as the connector gadgets have width 4.
Such a situation is shown to the right, where all initial variable vertices are to the right of the initial clause (the fat vertex).}
\label{fig_1_in_3}
\end{figure}

\subsection{Exactly three distinct variables per clause}
If we require the formula to have exactly three distinct variables in each clause, the reduction can also be modified accordingly.
Mansfield~\cite{mansfield} showed how to extend Lichtenstein's construction to obtain a planar 3-SAT formula with exactly three different variables per clause%
\footnote{This restriction has been re-discovered, e.g., in~\cite{hunt}.}
by constructing such a formula with planar incidence graph and a variable that is false in every satisfying assignment.

\begin{theorem}\label{thm:linked_three_distinct}
The \textsc{Linked Planar 3-SAT} problem remains NP-complete even if each clause contains exactly three different variables.
\end{theorem}
\begin{proof}
We reduce from the non-linked version of the problem by drawing an incidence graph on our grid and replace edge parts by gadgets that transport the truth settings of a variable.
We start with an $R$-drawing and extend it to a new $R$-drawing, similar to the proof of Theorem~\ref{thm:planar_1_in_3};
we add four pairs of clause and variable segments between each initial variable segment and the clause segment to its right, and three pairs of variable and clause segments between each initial clause segment and the variable segment to its right.
Thus, the number of clause segments strictly between a variable vertex $\tilde v$ and a clause vertex $\tilde c$ to its right is $4k$, for some natural number $k$, and if $\tilde c$ is to the left of $\tilde v$, then the number of clause segments strictly between them is $4k-1$.
Again, our connector gadgets have width 4, i.e., they connect two variables with two clause segments in-between them.
They are shown in \figurename~\ref{fig_three_distinct_1}.
The underlying formula of a single gadget connecting variables $x$ and $x'$ is
\begin{multline*}
(\neg x \lor a \lor u) \land (x' \lor \neg a \lor u) \land
(\neg x \lor b \lor \neg u) \land (x' \lor \neg b \lor \neg u) \land \\
(\neg x' \lor a' \lor u') \land (x \lor \neg a' \lor u') \land
(\neg x' \lor b' \lor \neg u') \land (x \lor \neg b' \lor \neg u')
\end{multline*}
By their construction, two variables $x$ and $x'$ always have the same value:
The part $(\neg x \lor a \lor u) \land (\neg a \lor u \lor x')$ entails $(\neg x \lor u \lor x')$, and $(\neg x \lor b \lor \neg u) \land (\neg b \lor \neg u \lor x')$ entails $(\neg x \lor \neg u \lor x')$.
These two clauses have $(\neg x \lor x')$ as a resolvent.
Analogously, the clauses containing $u'$ imply $x' \Rightarrow x$.

We note that in Mansfield's reduction it is no longer shown that there exists a variable cycle, so we cannot rely on a three-legged embedding.
As in the previous reduction, the connector gadgets have width 4.
This time, however, when all variable vertices of a clause vertex are to its right, one sequence of edge connectors replacing an edge of the initial incidence graph needs to take a larger detour to connect to the clause, as shown in \fig{fig_three_distinct_2}.
Doing this detour is possible as there is no vertex of the initial graph of distance 5 or less to the left of the clause vertex (recall that there are at least four clause segments strictly between a variable and a clause to its right).
As this results in a plane \textsc{Linked Planar 3-SAT} instance that is satisfiable if and only if the initial one was, the theorem follows.
\end{proof}

\begin{figure}
\centering
\includegraphics{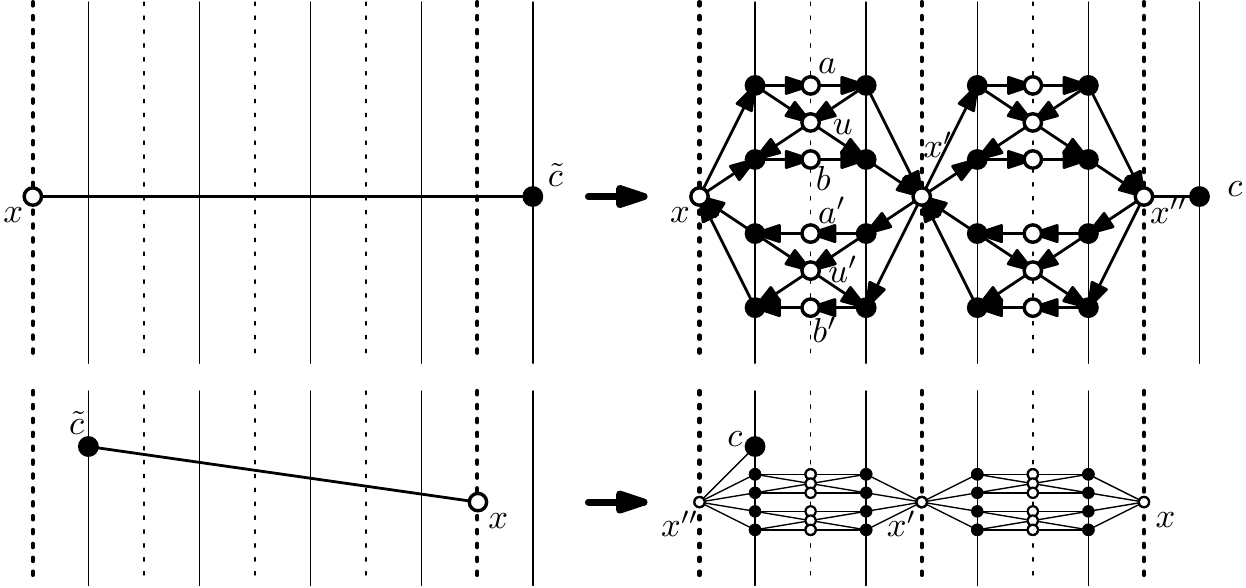}
\caption{A modified connector gadget.
A single gadget connects the variable vertices $x$ and $x'$.
(Hence, a single gadget has width~4.)
The variable $x$ has to have the same value as~$x'$ and $x''$.
All clauses have exactly three distinct variables.
Recall that an arrow from a variable to a clause indicates that the variable occurs negated;
for an arrow from the clause to the variable, the occurrence is unnegated.
(At the bottom, we omit arrows for clarity.)
}
\label{fig_three_distinct_1}
\end{figure}

\begin{figure}
\centering
\includegraphics[page=2]{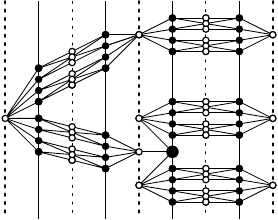}
\caption{
Since the connector gadgets have width 4, we may have to ``overshoot'' by 5 units when all variable vertices of an initial clause vertex (fat dot) are to its right.
As all vertices of the initial graph to the left of the clause vertex have distance at least 9, this construction is possible.
}
\label{fig_three_distinct_2}
\end{figure}

Note that the bi-implication implemented by the connector gadget has sixteen different truth assignments, independent of whether $x$ is true or false.
As the \textsc{Planar 3-SAT} problem is known to be \#P-complete even if each clause contains exactly three distinct variables~\cite{hunt}, we can add connector gadgets to transform any such instance into a \textsc{Linked Planar 3-SAT} instance.
For each connector gadget, the number of solutions is multiplied by~16.

\begin{theorem}
The \textsc{Linked Planar 3-SAT} problem remains \#P-complete even if each clause contains exactly three different variables.
\end{theorem}

\subsection{Monotonicity restrictions}
We can add the following list of restrictions to the setting for which the problem remains hard.
Both can be shown by reducing from \textsc{Monotone Planar 3-SAT}.

A clause is \emph{monotone} if it contains either only negated or only unnegated literals;
a formula is \emph{monotone} if all its clauses are monotone.

\begin{definition}[\textsc{Monotone Planar 3-SAT}]
Let $G_\phi = (C \cup V, E)$ be the incidence graph of a 3-SAT formula~$\phi$, where $C$ is the set of clauses and $V$ is the set of variables of~$\phi$, and each clause of $C$ is monotone.
The \textsc{Monotone Planar 3-SAT} problem asks whether $\phi$ is satisfiable.
\end{definition}

This problem has been shown to be NP-complete by de~Berg and Khosravi~\cite{de_berg};
they actually show that the problem remains hard even if there is a variable cycle separating the clauses with the negated variables from the ones with the unnegated variables.
That is, the incidence graph can be drawn in a rectilinear way with the variables on the $x$-axis such that all clauses with the negated occurrences are below the $x$-axis, and all other clauses are above the $x$-axis.
We can use their result to show hardness of the specialization of the problem to our setting.

\begin{theorem}\label{thm:monotone_linked}
The \textsc{Linked Planar 3-SAT} problem remains NP-complete even if all clauses are monotone.
\end{theorem}
\begin{proof}
We reduce from \textsc{Monotone Planar 3-SAT}.
The construction is the same as in the previous reductions, except for the connector gadget.
Again, we have a connector gadget of width 4 (described below).
Therefore, we start by an $R$-drawing and place four pairs of clause and variable segments between each variable segment and the clause segment to its right, as well as three pairs of variable and clause segments between each clause segment and the variable segment to its right (as in the reduction of Theorem~\ref{thm:linked_three_distinct}).

The connector gadget, shown in \fig{fig_monotone_connector}, consists of three variables $x$, $\overline{x}$, and $x'$.
We add two clauses $(x \lor \overline{x})$ and $(\neg x \lor \neg \overline{x})$, which entail $x \neq \overline{x}$.
Another pair of clauses, $(\overline{x} \lor x')$ and $(\neg \overline{x} \lor \neg x)$, imply $\overline{x} \neq x'$, and we thus have $x = x'$.
Again, if $x$ is to the left of a clause $\tilde c$ in the $R$-drawing, the number of clause segments strictly between $x$ and $\tilde c$ is a multiple of four, and so is the number of variable segments;
hence, we can place the gadgets as on the top of \fig{fig_monotone_connector}.
If $x$ is to the right of $\tilde c$, we place the connector gadget at $\tilde c$ as shown at the bottom of \fig{fig_monotone_connector}.

If all three variables of a clause are to its right, we take a detour in the same way as in the proof of Theorem~\ref{thm:linked_three_distinct};
see \figurename~\ref{fig_monotone_connector_2}.
\end{proof}

\begin{figure}
\centering
\includegraphics{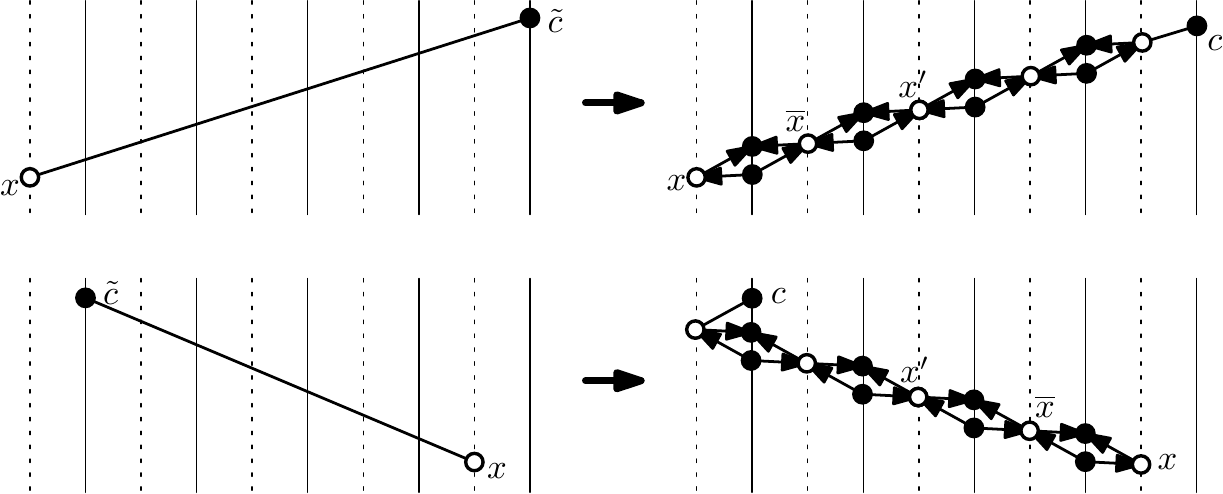}
\caption{A variant of the connector gadget in which all clauses are monotone.
We have $x \neq \overline{x} \neq x' = x$ to the right.
}
\label{fig_monotone_connector}
\end{figure}

\begin{figure}
\centering
\includegraphics[page=2]{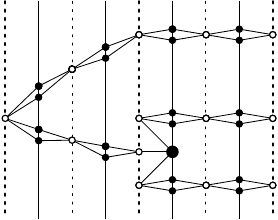}
\caption{As in the previous reductions, we use the above construction if all three variables of a clause are to its right.
}
\label{fig_monotone_connector_2}
\end{figure}

We note that the connector gadget of the previous reduction can be further split to show NP-hardness for instances in which a variable occurs only three times.
To refrain from being too repetitive, we only give a sketch of the construction in \fig{fig_monotone_three_occurrences}.
The variables $x$, $x'$, and $x''$ can be shown to have the same value by applying resolution:
for example, the clauses $(\neg x \lor \neg \overline{x}') \land (\overline{x}' \lor \overline{\overline{x}}') \land (\neg \overline{\overline{x}}' \lor \neg \overline{\overline{\overline{x}}}') \land (\overline{\overline{\overline{x}}}' \lor x')$ can be resolved to $(\neg x \lor x')$.
The remaining clauses of the cycle imply $(\neg x' \lor x)$.
In general, every second variable in the cycle will have the same value;
in the drawing, we have $x = x' = x''$.

\begin{figure}
\centering
\includegraphics[page=3]{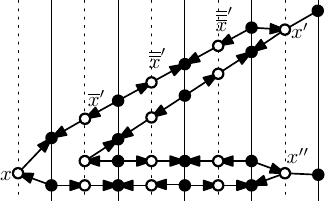}
\caption{The gadget shown in \fig{fig_monotone_connector} can be further split to have each variable vertex of degree at most three.}
\label{fig_monotone_three_occurrences}
\end{figure}

Let us now consider a variant that is of particular interest for the motivation of platform games mentioned in the introduction.
Lichtenstein already showed that \textsc{Planar 3-SAT} remains NP-complete even if we require that the variable cycle partitions the edges of every variable vertex into those leading to a negated occurrence and those leading to an unnegated one~\cite[Lemma~1]{lichtenstein}.
(As he mentions, this also implies that one could split a variable vertex into two literal vertices while preserving planarity.)
Note that there is a subtle difference to the restriction of de~Berg and Khosravi~\cite{de_berg} to \textsc{Monotone Planar 3-SAT}, as the side of the variable cycle to which the edges to the negated occurrences emanate is not fixed globally.
The following set of related restrictions could be particularly interesting for further reductions.

\begin{theorem}\label{thm_side}
The \textsc{Linked Planar 3-SAT} problem remains NP-complete even if, for each clause, the edges corresponding to positive occurrences emanate to the interior of the cycle~$\kappa$, and the ones to negated occurrences to the exterior.
In addition, each variable occurs in at most three clauses.
\end{theorem}
\begin{proof}
We reduce from the \textsc{Monotone Planar 3-SAT} variant of de~Berg and Khosravi~\cite{de_berg}, using a three-legged embedding.
The construction (see \figurename~\ref{fig_monotone}) uses cycles consisting of the variables $x_1, \dots x_k$ and $\overline{x_1}, \dots, \overline{x_k}$ and clauses
$(\neg x_i \lor x_{i+1})$ and $(\neg \overline{x_i} \lor \overline{x_{i+1}})$ for $i \in \{1, \dots, k-1\}$,
as well as the clauses $(x_1 \lor \overline{x_1})$ and $(\neg x_k \lor \neg \overline{x_k})$.
The latter thus entails $(\neg x_1 \lor \neg \overline{x_1})$.
We therefore have $x_i = x_j$, $\overline{x_i} = \overline{x_j}$, and $x_i \neq \overline{x_j}$ for $i,j \leq k$.
The variable vertices are placed from left to right with increasing indices.
Thus, the variables on the upper part of the cycle have the opposite value of those on the lower part.
For a clause with positive literals in the \textsc{Monotone Planar 3-SAT} instance, we connect the variables as in \figurename~\ref{fig_monotone}.

We therefore have variable gadgets that consist of mentioned cycles; see the three cycles at the bottom of \fig{fig_monotone}.
These variable gadgets are connected to the clause $c$ (which represents the clause of the initial formula) either directly (for the middle variable gadget), or via another cycle (for the left and the right variable gadget).
For these connector gadgets, we again have that the variables on the upper part have the opposite value as those of the lower part; in the way a variable gadget and a connector gadget are connected, a variable on the upper part of the variable gadget has the same value as one on the lower part of the connector gadget.
The clause $c$ has two edges to the left (which therefore correspond to negated occurrences) and one to the right (an unnegated occurrence).
In the setting shown in \fig{fig_monotone}, we consider a variable set to true if the variables on the lower part of the variable gadget are true.
A clause $\tilde c = (x \lor y \lor z)$ is thus transformed to $c = (\neg \overline{x_i} \lor \neg \overline{y_j} \lor z_k)$.
We can have variable gadgets of arbitrary width simply by having larger cycles, and connect the clauses according to the three-legged embedding.
The resulting construction is thus crossing-free, and the formula is satisfiable if and only if the initial one is.
\end{proof}

\begin{figure}
\centering
\includegraphics{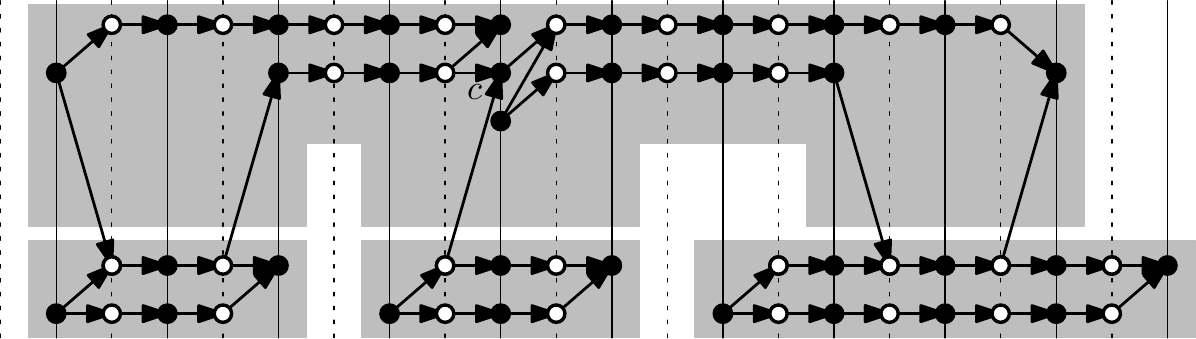}
\caption{All clauses have negative literals to the left and positive literals to the right.
The variable state is transported by cycles where the variables on top have the negative value of the variables at the bottom.
A clause $\tilde c = (x \lor y \lor z)$ is transformed to $c = (\neg \overline{x_i} \lor \neg \overline{y_j} \lor z_k)$.
The three-legged embedding of the clause is indicated by the gray contour.
}
\label{fig_monotone}
\end{figure}

Let us point out that we cannot strengthen the theorem by requiring that all clauses are monotone.
In that case, the clause vertices would be incident to the same face after the removal of $\kappa$.
The incidence graph then would be 2-outerplanar.
As discussed in the introduction, such instances can be solved in polynomial time.

Observe that the connector gadgets in \figurename~\ref{fig_monotone} use only clauses with two variables.
The reduction therefore also works for \textsc{Planar 2-SAT} instances.
While \textsc{Planar 2-SAT} can be solved in polynomial time, it is known to be \#P-complete~\cite{vadhan}.
As our gadgets are parsimonious (i.e., do not change the number of solutions), they can be applied to show \#P-completeness.

\begin{theorem}
The \textsc{Linked Planar 2-SAT} problem is \#P-complete.
It remains \#P-complete even if for each clause, the edges corresponding to positive occurrences emanate to the interior of the cycle~$\kappa$, and the ones to negative literals to the exterior.
\end{theorem}

In order to explore the boundaries of NP-hardness of satisfiability problems, not only restrictions on the number of occurrences per clause have been considered, but also on the number of times a variable occurs (see the survey~\cite{tippenhauer} for an overview).
In this connection, there are interesting effects when requiring exactly three different variables per clause.
Indeed, an NP-hardness reduction for the case where each variable occurs in at most three clauses requires that there are clauses with at most two (different) literals, as every CNF formula with exactly three literals per clause and at most three occurrences per variable is satisfiable~\cite{tovey}.
Darmann, D\"ocker, and Dorn~\cite{darmann_planar,darmann_journal} showed how to reduce the number of times a variable occurs.
For the variant of \textsc{Monotone Planar 3-SAT} that requires exactly three different variables per clause, the complexity was previously unknown~\cite{darmann_planar,darmann_journal}.
Surprisingly, it turns out that such instances are always satisfiable, as discussed in Section~\ref{sec_forcing_sat}.

\subsection{Remark: different cycles through clauses and~variables}\label{sec:different_cycles}
Observe that, for all our constructions, $G_\phi$ still allows for adding a variable cycle $H$, as well as a clause cycle $H'$, as shown in \figurename~\ref{fig_grid_region_2}.
But these cycles will, in general, cross mutually.
Also, they will cross the cycle~$\kappa$.
Recall that if $H$ did not cross $H'$ or $\kappa_C$, the problem would be solvable in polynomial time~\cite{demaine_lecture}.

\begin{figure}
\centering
\includegraphics[page=2]{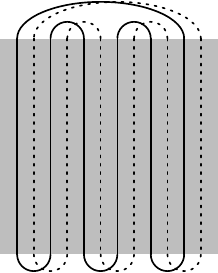}
\caption{A construction similar to the one in \fig{fig_grid_region} can be used to show that all the incidence graphs we obtain can be augmented by either a clause cycle or a variable cycle (and not just two paths obtained from~$\kappa$).}
\label{fig_grid_region_2}
\end{figure}

Using the gadgets and the two cycles shown in \figurename~\ref{fig_grid_region_2}, we observe that \textsc{Planar 3-SAT} remains NP-complete even if these cycles exist for all variants mentioned.
While the variable and clause cycles have been identified in~\cite{lichtenstein} and~\cite{clause_cycle}, respectively, it seems to have been unknown for the variants using exactly three variables per clause (even though Mansfield's construction~\cite{mansfield} can be embedded %
to obtain the cycles).
Also, it seems that clause cycles have not been considered for \textsc{Monotone Planar 3-SAT}.
Chaplick et al.~\cite{few_lines} showed hardness for \textsc{Planar Positive 1-in-3-SAT} with a clause cycle.

\section{Properties forcing satisfiability}\label{sec_forcing_sat}
Planarity is a rather drastic combinatorial restriction on the structure of a graph.
While NP-completeness of 3-SAT is preserved in the planar setting, further properties may lead not only to polynomial-time algorithms (as for \textsc{Planar NAE-SAT}~\cite{naesat}), but also to instances that are always satisfiable.
The following statement does not only hold for \textsc{Planar 3-SAT}, but for an arbitrary number of variables per clause.

\begin{theorem}\label{thm:three_distinct_satisfiable}
Every instance of \textsc{Planar SAT} in which each clause has at least three negated or at least three unnegated occurrences of distinct variables is satisfiable.
A satisfying assignment can be found in quadratic time.
\end{theorem}
\begin{proof}
Consider any plane drawing of the incidence graph of the formula.
For each clause vertex, we can add a 3-cycle consisting of three of its variable vertices that are either all negated or all unnegated,  as shown in \figurename~\ref{fig_monotone_distinct}.
(For clauses with three variables this is similar to a Y-$\Delta$ transform, a common operation to replace a vertex of degree 3 by a 3-cycle:
connect two variables by a curve in a neighborhood of the two edges connecting them to the clause.)
Observe that, after removing the clause vertices, the resulting graph is still planar.
It is thus 4-colorable~\cite{fourcolor} and we may consider any 4-coloring of the variables.
Set the variables of the vertices with colors $1$ and $2$ to true, and the others to false.
A 3-cycle contains three different colors, and thus a 3-cycle through tree (monotone) variable vertices has at least one variable set to true and one variable set to false.
\end{proof}

\begin{figure}
\centering
\includegraphics{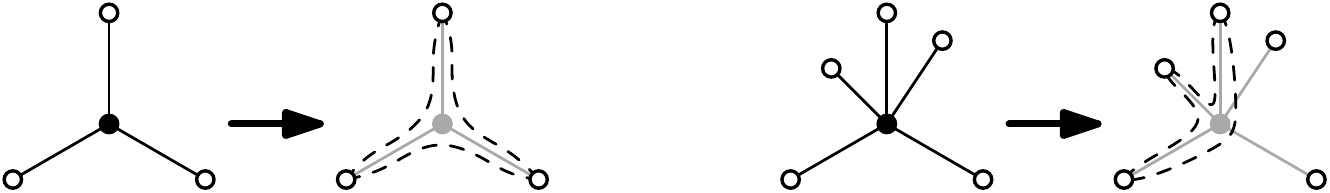}
\caption{Clause (black) with at least three distinct variables (white).
We can augment the graph with a 3-cycle through variables (dashed) and remove the clauses while preserving planarity.
In any 4-coloring of the graph, the three variables belong to three different color classes and thus there always exists a satisfying truth assignment by setting the variables of two color classes to true.}
\label{fig_monotone_distinct}
\end{figure}

\begin{corollary}
Every instance of \textsc{Monotone Planar SAT} with at least three distinct variables per clause is satisfiable.
\end{corollary}

\begin{corollary}\label{cor_five_var}
Every instance of \textsc{Planar SAT} with at least five distinct variables per clause is satisfiable.
\end{corollary}

For at least four variables per clause, we can give a similar result, closing the gap between Corollary~\ref{cor_five_var} and Mansfield's result~\cite{mansfield}, using less heavy machinery than the Four-Color theorem.
The proof makes use of Hall's theorem~\cite{hall}, inspired by a technique by Tovey~\cite{tovey}: each subset of $k$ clauses has at least $k$ variables occurring in it.

\begin{lemma}\label{lem_matching}
Let $G = (B \cup C, E)$ be a bipartite planar graph with parts $B$ and $C$ such that all vertices in $C$ have degree at least four.
Then there exists a matching covering every vertex of $C$.
\end{lemma}
\begin{proof}
Consider any plane embedding of $G$ and let $C' \subseteq C$ be an arbitrary subset of $C$.
Then we denote by $B' \subseteq B$ the elements of $B$ that are adjacent to an element of~$C'$.
Let $G'$ be the subgraph of $G$ induced by $C' \cup B'$, let $E'$ be the set of edges of $G'$, and let $F'$ be the set of the faces of $G'$.
Further, let $k \geq 1$ be the number of connected components of $G'$.
Euler's formula states that $|B'|+|C'|-|E'|+|F'| = 1 + k$.
Consider any edge $\{b, c\}$ and the two ways $(b,c)$ and $(c,b)$ of directing it.
To every directed edge, we assign the face that is directly to the left of this edge.
Hence, every face is assigned to each of the directed edges along its boundary.
Since $G'$ is bipartite, every face is assigned to at least four directed edges.
We get
\begin{equation}\label{eqn_faces}
2|E'| = \sum_{f \in F'} |f| \geq 4|F'| \enspace , %
\end{equation}
where $|f|$ is the number of directed edges along the boundary of the face~$f$.
(Note that an undirected edge may have same face~$f$ on both of its sides;
such an edge counts twice for~$|f|$.)
Combining Euler's formula and~(\ref{eqn_faces}) to
$
|E'| \geq 2(1 + k - |B'| - |C'| + |E'|)
$
gives the bound
\begin{equation}\label{eqn_edge_upper}
2 |B'| + 2 |C'| - 2 - 2k \geq |E'| \enspace .
\end{equation}
Further,
$|E'| \geq 4|C'|$, as $G'$ is bipartite and every element of $C'$ is incident to at least four edges.
Combining this with~(\ref{eqn_edge_upper}) results in
\begin{equation}\label{eqn_neighborhood}
|C'| \leq |B'| - 1 - k \enspace .
\end{equation}
Recall that $C'$ is an arbitrary subset of $C$, and thus, for every subset of $C$, the set of adjacent elements in $B$ has at least the same cardinality.
Therefore, by Hall's theorem~\cite{hall}, there is a matching on $G$ that covers~$C$.
\end{proof}

\begin{theorem}\label{thm:four_satisfiable}
Every CNF formula with planar incidence graph of $n$ vertices and at least four distinct variables per clause has a satisfying assignment, which can be found in $O(n^{1.5})$ time.
\end{theorem}
\begin{proof}
Let $\phi$ be a planar SAT instance with vertex set $V$ and clause set $C$, and let $G_\phi = (V \cup C, E)$ be the associated incidence graph.
By Lemma~\ref{lem_matching}, there is a matching on $G_\phi$ covering~$C$, which assigns a distinct variable to each clause (i.e., a system of distinct representatives for the clauses).
If we set the corresponding literal to true, we get an assignment satisfying~$\phi$.
The matching can be found using the algorithm by Hopcroft and Karp~\cite{hopcroft_karp}, which runs in $O(|E|\sqrt{v})$.
As $G$ is planar, $|E| \in O(n)$, from which the claimed time bound follows.
\end{proof}

\paragraph*{Acknowledgments.} The \textsc{Linked Planar 3-SAT} problem has been brought to our attention by Lena Schlipf.
We thank Erik Demaine for pointing out the motivation in platform games, as well as Patrick Schnider, Michael Hoffmann, and Janosch D\"ocker for valuable discussions.
We also thank an anonymous reviewer for the indicated simplification in the proof of Theorem~\ref{thm:planar_1_in_3}.

\bibliographystyle{plainurl}
\bibliography{bibliography}

\end{document}